\documentclass[12pt]{amsart}

\usepackage[a4paper,margin=1in]{geometry}

\usepackage[utf8]{inputenc}
\usepackage[T1]{fontenc}

\usepackage[english]{babel}

\usepackage{amsmath}
\usepackage{graphicx}
\usepackage{amssymb}
\usepackage[colorlinks=true, allcolors=blue]{hyperref}

\newtheorem{theorem}{Theorem}
\newtheorem{lemma}[theorem]{Lemma}
\newtheorem{conj}{Conjecture}

\title{A note about majority colorings of countable DAGs}
\author{Bartłomiej Bosek$^{*}$}
\thanks{$^{*}$ Research supported by the National Science Center of Poland under grant no. 2020/37/B/ST1/03298.}
\email{bartlomiej.bosek@uj.edu.pl}
\author{Aleksander Katan}
\email{aleksander.a.katan@gmail.com}
\address{Jagiellonian University\\
Faculty of Mathematics and Computer Science\\
Theoretical Computer Science Department\\
ul. {\L}ojasiewicza 6, Kraków, Poland}

\begin{document}

\begin{abstract}
A majority coloring of an undirected graph is a vertex coloring in which for each vertex there are at least as many bi-chromatic edges containing that vertex as monochromatic ones. It is known that for every countable graph a majority 3-coloring always exists. The Unfriendly Partition Conjecture states that every countable graph admits a majority 2-coloring. Since the 3-coloring result extends to countable DAGs, a variant of the conjecture states that 2 colors are enough to majority color every countable DAG. We show that this is false by presenting a DAG for which 3 colors are necessary.

Presented construction is strongly based on a StackExchange conversation \cite{stackLabellings} regarding labellings of infinite graphs. 

\end{abstract}

\maketitle

\section{Introduction}

Given a finite simple directed graph $G=(V, E)$, a coloring $\phi : V(G) \rightarrow \mathbb{N}$ is called a majority coloring if for every vertex $v \in V(G)$

$$\#\{vu \in E(G) : \phi(u) \neq \phi(v)\} \geqslant \#\{vu \in E(G) : \phi(u) = \phi(v)\}.$$

In other words, for every vertex $v$ at most half of the out-edges from $v$ may be monochromatic. A digraph $G$ is majority $k$-colorable if it admits a majority coloring using at most $k$ colors.

It is a well known fact that every simple finite DAG (directed acyclic graph) admits a majority 2-coloring. It can be constructed by sorting the graph topologically, and then coloring the vertices greedily in the reverse topological order.

The definition extends directly to graphs of any cardinality. In \cite{anholcer2020majority} Anholcer, Bosek and Grytczuk prove that every countable DAG is not only majority 4-colorable, but majority 4-choosable. In \cite{Haslegrave_2020} Haslegrave improves their result and shows that every countable DAG is majority 3-choosable, which implies majority 3-colorability. In both \cite{anholcer2020majority} and \cite{Haslegrave_2020} it is conjectured that every countable DAG is majority 2-colorable. Here, we give a counterexample to this.

\section{The counterexample}

\subsection{Construction of the counterexample}

\begin{figure}
    \centering
    \includegraphics[width=0.8\linewidth]{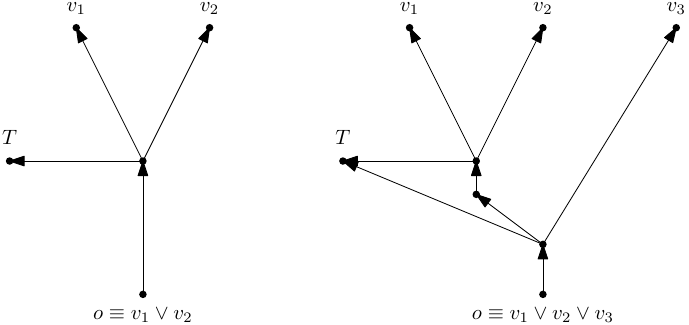}
    \caption{\label{fig:or_gadgets}An OR$^T(v_1, v_2)$ gadget and an OR$^T(v_1, v_2, v_3)$ gadget.}
\end{figure}

Our construction makes use of gadgets called by us the OR gadgets. Intuitively, given a graph $G$ and a vertex $T \in V(G)$, one can interpret a 2-coloring of $G$ as as a truth assignment to the vertices of $G$ based on whether the vertex is colored the same color as $T$ or not. Given $v_1, v_2 \in V(G) - \{T\}$, an OR$^T(v_1, v_2)$ gadget is as defined as shown in Figure \ref{fig:or_gadgets}. \\
Let us refer to the vertex $o$ in Figure \ref{fig:or_gadgets} as an output vertex of the gadget. \\
Let OR$^T(v_1, \dots, v_k)$ denote a chained OR gadget, like exemplified in Figure \ref{fig:or_gadgets}. Output $o$ of such gadget is defined analogously,

Let us now present the counterexample graph $G$. As shown in Figure \ref{fig:counterexample}, it is constructed by taking an infinite directed path $(v_1, v_2, v_3, \dots)$, a vertex $T$, for every $i, j \in \mathbb{N}$ such that $2 \leq i < j$ an OR$^T(v_i, \dots, v_j)$ gadget (later referred to as OR$_{i, j}$), and finally joining the output of OR$_{i, j}$ with an edge starting in $v_{i-1}$.

\begin{figure}
    \centering
    \includegraphics[width=0.8\linewidth]{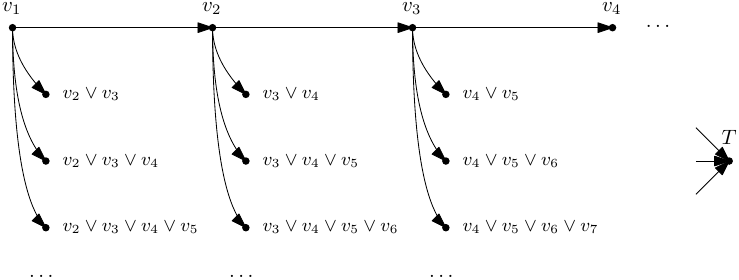}
    \caption{\label{fig:counterexample}A countable acyclic digraph $G$ that we claim not to be majority 2-colorable.}
\end{figure}

\subsection{Proof that 3 colors are required for any majority coloring of $G$.}

An OR$^T(u_1, \dots, u_k)$ gadget consisting of vertices $U$ is considered \emph{valid} in $H$ if there are no outgoing edges between $U - \{T, u_1, \dots, u_k\}$ and $V(G) - U$ in $H$. Note that every OR$_{i, j}$ is valid in $G$. \\
For any 2-coloring $\phi$ of $H$ and any $v, T \in V(G)$, let $\phi^T(v)$ be true iff $\phi(v) = \phi(T)$. We will say that $v$ is colored \emph{true} in $\phi$.

\begin{lemma}\label{lem:outputs}
Given a graph $H = (V, E)$ of any cardinality and any majority 2-coloring $\phi$ of $G$, if $U$ is a valid OR$^T(u_1, \dots, u_k)$ gadget in $H$ with an output vertex $o$, then $\phi^T(u_1) \lor \dots \lor \phi^T(u_k) = \phi^T(o)$.
\end{lemma}

\begin{proof}
Induction on $i$. 

Base case: since there are no outgoing edges from $U - \{T, u_1, u_2\}$ to the rest of the graph, it suffices to check all 4 possible precolorings of $v_1$ and $v_2$ in regard to the color of $T$ and see that in all possible extensions to $U$ the output vertex is colored as stated. 

Induction step: note that OR$^T(u_1, \dots, u_l)$ is in fact an OR$^T(o, u_l)$ on the output $o$ of OR$^T(u_1, \dots, u_{l-1})$. The same reasoning applies.
\end{proof}

\begin{lemma}\label{lem:constraint}
Let $\phi$ be a majority 2-coloring of $G$. For each $i \in \mathbb{N}_{+}$ $\phi^T(v_i)$ is true if and only if for each $j>i$ $\phi^T(v_j)$ is not true.
\end{lemma}

\begin{proof}

Let $i \in \mathbb{N}_+$.

$\implies$

Contra-position. If there exists $j > i$ such that $\phi^T(v_j)$ is true, then notice, that the output of OR$_{i+1, k}$ is true for any $k \geqslant j$ as well by Lemma \ref{lem:outputs}. Therefore, among the out-neighbors of $v_i$, infinitely many are colored true, and finitely many are colored false, thus $\phi^T(v_i)$ must be false.

$\impliedby$

If all $v_j$ are colored false for $j > i$, then all outputs of OR$_{i, j}$ must as well be colored false by Lemma \ref{lem:outputs} and therefore $v_i$ must be colored true.

\end{proof}

\begin{theorem}
The graph $G$ is a simple countable acyclic digraph that is not majority 2-colorable. 
\end{theorem}

\begin{proof}

$G$ is trivially countable. 

$G$ is acyclic. Note that the vertex $T$ has out-degree $0$. We will show a topological sort $\sigma$ of the graph $G-\{T\}$. To each vertex $v$ we will assign a triplet $\sigma(v) = (i, j, k) \in \mathbb{N}^3$ so that for every edge $uv$ of $G$ it is true that $\sigma(u) <_L \sigma(v)$, where $<_L$ denotes the lexicographical ordering. 

Note that each OR$_{i, j}$ used in the construction is acyclic. Let $\sigma_{i, j} : V(\text{OR}_{i, j}) \rightarrow \mathbb{N}$ denote the topological sort of OR$_{i, j}$. 

We assign the triplets as follows:

\begin{itemize}
    \item $\sigma(v_i) = (i, 0, 0)$,
    \item for $v \in V(\text{OR}_{i, j}) - \{v_i : i \in \mathbb{N}_+\}$, $\sigma(v) = (i-1, j, \sigma_{i, j}(v))$.
\end{itemize}

For every edge type, we show that the inequality holds:

\begin{itemize}
    \item for $v_iv_{i+1} \in E(G-\{T\})$, we have 
    
    $\sigma(v_i) = (i, 0, 0) <_L (i+1, 0, 0) = \sigma(v_{i+1})$,
    \item for $v_io$ where $o$ is an output of OR$_{i+1, j}$, we have

    $\sigma(v_i) = (i, 0, 0) <_L (i, j, \sigma_{i+1, j}(o)) = \sigma(o)$,

    \item for $uv_k$ where $u \in V(\text{OR}_{i, j}) - \{v_i : i \in \mathbb{N}\}$ and $k \in [i, j]$, we have

    $\sigma(u) = (i-1, j, \sigma_{i, j}(u)) <_L (k, 0, 0) = \sigma(v_k)$,

    \item for $uw$ where $u, w \in V(\text{OR}_{i, j}) - \{v_i : i \in \mathbb{N}\}$, due to the definition of $\sigma_{i, j}$ we have

    $\sigma(u) = (i-1, j, \sigma_{i, j}(u)) <_L (i-1, j, \sigma_{i, j}(w)) = \sigma(w)$.
\end{itemize}

Every edge of $G-\{T\}$ was considered, and $T$ is a vertex of outdegree $0$, therefore $G$ is acyclic.

Let us assume that a majority 2-coloring $\phi$ of $G$ exists - we will eventually reach a contradiction. \\
From Lemma \ref{lem:constraint} notice that there is at most one vertex among $v_1, v_2, \dots$ colored true. If exactly one vertex $v_k$ is colored true, notice that $v_{k+1}$ also satisfies the constraint of Lemma \ref{lem:constraint}, and therefore should also be colored true. 
Thus, all of $v_1, v_2, \dots$ must be colored false. But in this case, $v_1$ satisfies the constraint of Lemma \ref{lem:constraint}, and therefore should be colored true. \\
This gives us a contradiction to our assumption that a majority 2-coloring of $G$ exists.

\end{proof}

\section{Conclusions and open problems}

We state a problem that may act like a bridge between this result and the Unfriendly Partition Conjecture. 

By a multigraph, we mean an undirected graph $G=(V, E)$ together with a weight function $w : E(G) \rightarrow \mathbb{N}$. A coloring $\phi : V(G) \rightarrow \mathbb{N}$ is called a majority coloring of $G$ if for every vertex $v \in V(G)$

$$\sum_{vu \in E(G) \land \phi(v) \neq \phi(v)} w(u) \geqslant \sum_{vu \in E(G) \land \phi(v) = \phi(v)} w(u) $$

where sigma denotes the limit of a series in the countable case.

Just like finite undirected graphs, by a max-cut argument, finite multi-graphs are majority 2-colorable. Multi-graphs allow us to emulate DAGs to some degree by manipulating the weight function. 
Unfortunately, the presented construction for directed graphs does not directly transfer to multi-graphs.

Nevertheless, we postulate the following.

\begin{conj}    
There exists an undirected countable multi-graph that is not majority 3-colorable. 
\end{conj}

\section{Acknowledgments}

We thank Stack Exchange users \href{https://math.stackexchange.com/users/238653/jo-bain}{Jo Bain}, \href{https://math.stackexchange.com/users/111012/bof}{bof} and \href{https://math.stackexchange.com/users/71850/alex-ravsky}{Alex Ravsky} for providing a really strong inspiration for this result in \href{https://math.stackexchange.com/questions/1271698/labelings-of-infinite-directed-acyclic-graphs}{this StackExchange conversation} \cite{stackLabellings}.

\bibliographystyle{alpha}
\bibliography{references}

\end{document}